\documentclass[a4paper,10pt]{article}

\usepackage[english]{babel}
\usepackage[T1]{fontenc}
\usepackage[utf8]{inputenc}
\usepackage{amsthm}
\usepackage{a4wide}
\usepackage{authblk}

\usepackage{amsmath}
\usepackage{amssymb}
\usepackage{dsfont}     
\usepackage{tikz}       
\usepackage{graphicx}   
\usepackage{pifont}     

\newtheorem{definition}{Definition}[section]
\newtheorem{lemma}[definition]{Lemma}
\newtheorem{satz}[definition]{Theorem}

\newtheorem{theorem}[definition]{Theorem}

\newtheorem{folgerung}[definition]{Lemma}
\newtheorem{bemerkung}[definition]{Remark}

\newcommand{\obda}{\textnormal{without loss of generality}}

\newcommand{\ohnenull}{\backslash\{0\}}

\newcommand{\ylimit}{y_{\text{\textnormal{\tiny limit}}}}
\newcommand{\flimit}{f_{\text{\textnormal{\tiny limit}}}}

\newcommand{\TO}{\longrightarrow}

\newcommand{\Impl}{\Longrightarrow}
\newcommand{\equi}{\Leftrightarrow}
\newcommand{\Equi}{\Longleftrightarrow}

\newcommand{\Reell}{\mathbb{R}}

\newcommand{\M}{\mathbb{M}^{n\times n}}

\DeclareMathOperator{\Sym}{\mathbb{S}ym}
\DeclareMathOperator{\PSym}{\mathbb{PS}ym}

\newcommand{\abb}[5]{#1\hspace{1 mm}:\hspace{2 mm}#2\hspace{1 mm}\TO\hspace{1 mm}#3\hspace{1 mm},\hspace{2 mm}#4\hspace {1 mm}             \longmapsto\hspace{1 mm}#5} 

\newcommand{\alle}{\forall\,}

\newcommand{\allealle}[2]{\alle #1\;\;\alle #2}

\newcommand{\sprodukt}[2]{\left\langle #1\, ,\hspace*{0.04em} #2\right\rangle}

\DeclareMathOperator{\spn}{\text{\textnormal{span}}}
\DeclareMathOperator{\diag}{\text{\textnormal{diag}}}

\DeclareMathOperator{\Adj}{\text{\textnormal{Adj}}}
\DeclareMathOperator{\tr}{\text{\textnormal{tr}}}

\newcommand{\Id}{\mathds{1}} 

\newcommand{\Konvexbedingung}{\eqref{eq:konvexitaetsbedfuerproblem}}
\newcommand{\Aequizukonvexbedingung}{\eqref{eq:aequivalentzukonvexbed}}
\newcommand{\Lawp}{\eqref{eq:LAWP}}

\newcommand{\Dglugl}{\eqref{eq:dgluglfuerproblem}}
\newcommand{\Ref}[2]{\hyperref[#2]{#1 \ref*{#2}}}
\newcommand{\Refnr}[1]{\hyperref[#1]{\textnormal{(\ref*{#1})}}}
\newcommand{\Refnrohne}[1]{\eqref{#1}}
\newcommand{\Refname}[1]{\hyperref[#1]{\ref*{#1}}}

\newcommand{\my}{\mu}

\begin{document}
\title{On the convexity of the function $C \mapsto f(\det C)$ on positive definite matrices.}
\author[1]{Stephan Lehmich}
\author[1,2]{Patrizio Neff}
\author[1]{Johannes Lankeit}
\affil[1]{Lehrstuhl für Nichtlineare Analysis und Modellierung, Fakult\"at f\"ur Mathematik, Universit\"at Duisburg-Essen}
\footnotetext[2]{To whom correspondence should be addressed. e-mail: {\tt patrizio.neff@uni-due.de}}

\maketitle
\begin{abstract}
 We prove a condition on $f\in C^2(\Reell_+,\Reell)$ for the convexity of $f\circ\det$ on $\PSym(n)$, namely that $f\circ\det$ is convex on $\PSym(n)$ if and only if
 \[f''(s)+\frac{n-1}{ns}\cdot f'(s)\geqslant 0\quad\text{ and }\quad f'(s)\leqslant 0\qquad\alle s\in\Reell_+.\]
 This generalizes the observation that $C\mapsto -\ln \det C$ is convex as a function of $C$.
\end{abstract}

\section{Introduction}\label{chap:einleitung}
The question of how to choose physically reasonable strain energy functions in nonlinear elasticity has attracted much attention and is not yet completely solved. The major breakthrough came with John Ball's seminal contributions \cite{Ball77,Ball78,Ball02} introducing polyconvexity, i.e.\ convexity of the strain energy $W$ as a function of the arguments $(F,{\rm Cof}F,\det F)$, see also \cite{Schroeder_Neff01,Schroeder_Neff_Ebbing07}. Polyconvexity reconciles the physically reasonable growth condition $W(F)\to\infty$ as $\det F\to 0$ with the weak-lower-semicontinuity (quasiconvexity), which in return implies ellipticity.
A very simple example of a polyconvex function is the uni-constant compressible Neo-Hooke model
\[
 W_{NH}(F)=\my[\sprodukt{F^TF-\Id}{\Id}-2\ln\det(F)],\qquad \text{shear modulus }\my>0 .
\]
The strain energy is isotropic, frame-indifferent, polyconvex, i.e.\ convex as a function of $(F,\det F)$, stress-free in the reference configuration and $W_{NH}\to \infty$ as $\det F\to 0$. It is well known that the latter requirement excludes from the outset that $F\mapsto W_{NH}(F)$ may be a convex function of $F$, \cite{Neff_CISM07}. 
Howewer, rewriting $W_{NH}$ in terms of the Cauchy-Green deformation tensor $C=F^TF$, which gives
\[
 W_{NH}(F)=\widehat{W}_{NH}(C)=\my(\sprodukt{C-\Id}{\Id}-\ln\det C)
\]
one may readily check that $C\mapsto \widehat{W}_{NH}(C)$ is a convex function of $C$, despite its singularity in the determinant as $\det C\to 0$.\\
We surmise that convexity of the free energy w.r.t.\ $C$ (or the stretch tensor $U=\sqrt{C}$) is an additional, desirable feature of any free energy as it implies monotonicity of the stress-strain relation.\\
In this short contribution we therefore investigate, which functions $f\in C^2(\Reell_+,\Reell)$ are such that $C\mapsto f(\det C)$ is convex as function of $C\in \PSym(n)$ and generalize the 
well-known result that $C\,\mapsto -\ln \det C$ is convex on the set of positive definite symmetric matrices \cite{Strang1991,Magnus1999,Mirsky1972} 
by proving:

\begin{theorem}\textsc{(A differential inequality characterization)
}
\label{theo:dgluglfuerproblem}
  Let $f\in C^2(\Reell_+,\Reell)$. Then the function
  \[
   \abb{f\circ\det}{\PSym(n)}{\Reell}{C}{f(\det C)}
  \]
 is convex if and only if  
  \begin{equation}\label{eq:dgluglfuerproblem}
   \boxed{f''(s)+\frac{n-1}{ns}\cdot f'(s)\geqslant 0\quad\text{and}\quad f'(s)\leqslant 0\qquad\alle s\in\Reell_+\; .}
  \end{equation}  
\end{theorem}
\begin{proof}
 This is an immediate consequence of Lemmas \ref{lem:fmonotonfallend}, \ref{lem:unglfuerzweiteablnotwendig} and \ref{lem:hinreichend}.
\end{proof}
In the following we will reformulate the condition for convexity to obtain this result. We start with some preliminaries:

By $\M$ we denote the set of all $n\times n$-matrices, $\Sym(n)$ stands for the set of all real symmetric $n\times n$-matrices and $\PSym(n)$ for the set of all real symmetric positive definite $n\times n$-matrices.

\begin{lemma}\textsc{(Characterization of convexity)}
  Let $X$ be a normed space, $g\in C^2(K,\Reell)$ and $K\subseteq X$ open and convex. Then
  \begin{equation}
   g \text{ convex}\quad\Equi\quad D^2g(x).(z,z)\geqslant 0\quad\alle x\in K,\, z\in\spn(K)\; .
  \end{equation}
\end{lemma}
\begin{proof}
 \cite[p.27]{Rockafellar70}
\end{proof}
In particular we obtain
\begin{satz}\label{theo:konvexitaetsbed}\textsc{(Condition for convexity)}
  For $g\in C^2\left(\PSym\left(n\right),\Reell\right)$ we have
  \begin{equation}
   g \text{ convex}\quad\Equi\quad D^2g(C).(H,H)\geqslant 0\quad\alle C\in\PSym(n),\, H\in\Sym(n)\; .
  \end{equation}
\end{satz}
\begin{proof}
 Let \(K:=\PSym(n)\) and \(X:=\Sym(n)\) in the previous lemma. $\PSym(n)$ is an open convex subset of the normed space $\Sym(n)$ (with operator norm): 
 Use the characterization $A\in \PSym(n)\;\iff\;\sprodukt{Ax}{x}>0\,\alle x\in\Reell^n\ohnenull\;$ and for convexity also the Cauchy-Schwarz-inequality. Furthermore \(\spn(K)=\spn(\PSym(n))=\Sym(n)=X\):  
 The inclusion \glqq$\subseteq\,$\grqq\ is obvious. For the other inclusion write $A$ as a diagonal matrix (the corresponding transformation preserves positive definiteness and symmetry) and show that this can be written as a linear combination of positive definite symmetric matrices.
\end{proof}
\vspace*{\baselineskip}

By $\sprodukt{A}{B}=\tr(AB^T)$ we denote the trace inner product of the matrices $A$ and $B$.

\begin{satz}\label{theo:konvexitaetsbedfuerproblem}\textsc{(A condition for convexity)}

  Let $f\in C^2(\Reell_+,\Reell)$. Then the function
  \[
   \abb{g:=f\circ\det}{\PSym(n)}{\Reell}{C}{f(\det C)}
  \]
  is convex if and only if 
  \begin{equation}\label{eq:konvexitaetsbedfuerproblem}
   \boxed{\begin{gathered}   
    \allealle{C\in\PSym(n)}{H\in\Sym(n)}:\\[1ex]
    \Bigl[f''(\det C)\!\det C\!+\!f'(\det C)\Bigr]\left\langle C^{-1}, H\right\rangle^2- f'(\det C)\!\left\langle HC^{-1},                    C^{-1}H\right\rangle\geqslant 0.
   \end{gathered}}
  \end{equation}
\end{satz}
\begin{proof}
 Because $f\in C^2, \det\in C^\infty$, also $g\in C^2$. It remains to be shown that
 \begin{multline*}
  D^2g(C).(H,H)=\det C\cdot\Bigl\{\bigl[f''(\det C)\cdot\det C + f'(\det C)\bigr]\cdot\sprodukt{C^{-1}}{H}^2\\
  - f'(\det C)\cdot\sprodukt{HC^{-1}}{C^{-1}H}\Bigr\}
 \end{multline*}
 for $C\in\PSym(n)$ and $H\in\Sym(n)$, then the claim follows by Theorem \ref{theo:konvexitaetsbed}. 
 Because $\det$ is infinitely often differentiable on $\mathbb{M}^{n\times n}$ and (cf. \cite{Neff_CISM07})
   \(D\det(A).H=\sprodukt{\Adj A^T}{H}\), where $\Adj A$ denotes the adjugate matrix of $A$. For invertible $C$ and symmetric $H$ we have \(D\det(C).H=\det{C}\sprodukt{C^{-1}}{H}\), 
 and hence obtain by the chain rule
 \begin{align*}
  Dg(C).H = Df(\det C)D\det(C).H 
          = f'(\det C)\cdot\det C\cdot\sprodukt{C^{-1}}{H},
\end{align*}
and therefore, by chain rule and the fact that $D[C^{-1}].H=-C^{-1}HC^{-1}$,
 \begin{alignat}{2}
  D^2g(C).(H,H)&=&& f''(\det C)\cdot\left(\det C\right)^2\cdot\sprodukt{C^{-1}}{H}^2 
               \,+\, f'(\det C)\cdot\det C\cdot\sprodukt{C^{-1}}{H}^2 \notag\\
               &&\,+\,& f'(\det C)\cdot\det C\cdot\sprodukt{-C^{-1}HC^{-1}}{H} \\               
               &=&& f''(\det C)\cdot\det C^2\cdot\sprodukt{C^{-1}}{H}^2
               \,+\, f'(\det C)\cdot\det C\cdot\sprodukt{C^{-1}}{H}^2 \notag\\
               &&\,-\,& f'(\det C)\cdot\det C\cdot\sprodukt{HC^{-1}}{C^{-1}H} \; . \notag\tag*{\qedhere}
 \end{alignat}
\end{proof}

\begin{lemma}\label{lem:fmonotonfallend}
 The inequality \(f'(s)\leqslant 0\quad\alle s\in\Reell_+\,\)
  is necessary for \Konvexbedingung.
\end{lemma}
\begin{proof}
 Assume $s\in\Reell_+$ satisfying $f'(s)>0\,$. Let
  \(C=\diag(1,...,1,s)\in\PSym(n)\) and \(H=\diag(1,-1,0,0,...)\in\Sym(3)\). Then $\,\det C=s\,$, $\,\sprodukt{C^{-1}}{H}=1-1=0\,$ and
  \(\sprodukt{HC^{-1}}{C^{-1}H}=\sprodukt{\diag(1,-1,0,...,0)}{\diag(1,-1,0,...,0)}=2\). Together with \Refnrohne{eq:konvexitaetsbedfuerproblem} we obtain
  \(-2 f'(s)\geqslant 0\), a contradiction to $f'(s)>0$.
\end{proof}

\begin{lemma}\label{lem:aequivalentzukonvexbed}
  \Konvexbedingung\ holds if and only if
  \begin{equation}\label{eq:aequivalentzukonvexbed}
   \boxed{\begin{gathered}
    \allealle{H\in\Sym(n)}{D^{-1}=\diag(d_1,...,d_n)}\; ,\\
    \text{where $\; d_1,...,d_n\in\Reell_+\quad\text{and}\quad s^{-1}:=\det D^{-1}=d_1\cdot\ldots\cdot d_n\in\Reell_+:$}\\[1ex]
    \left(f''(s) + \frac{f'\left(s\right)}{s}\right)\sprodukt{D^{-1}}{H}^2 - \frac{f'(s)}{s}\sprodukt{D^{-1}H}{HD^{-1}}\geqslant 0\; .  
   \end{gathered}}
  \end{equation}  
\end{lemma}
\begin{proof}
 Consider an arbitrary $C\in\PSym(n)$ in \Refnrohne{eq:konvexitaetsbedfuerproblem}. Then there is an orthogonal matrix $Q$, such that \(
  C=QDQ^T\quad\equi\quad C^{-1}=QD^{-1}Q^T
 \),
 where  $D=\diag(\lambda_1,...,\lambda_n)$ and positive $\lambda_i$. 
By the properties of the scalar product of matrices we have
\[
\sprodukt{C^{-1}}{H}=\sprodukt{QD^{-1}Q^T}{H}=\sprodukt{QD^{-1}}{HQ}=\langle D^{-1},Q^THQ\rangle.
\]
For $H\in \Sym(n)$ let $\widetilde{H}:=Q^THQ\) and note that $H$ varies over the whole of $\Sym(n)$ if and only if $\widetilde{H}$ does. 
Analogously,
\[  \sprodukt{HC^{-1}}{C^{-1}H}=\sprodukt{HQD^{-1}Q^T}{QD^{-1}Q^TH}\\
                              =\sprodukt{\widetilde{H}D^{-1}}{D^{-1}\widetilde{H}}\\
                              =\sprodukt{D^{-1}\widetilde{H}}{\widetilde{H}D^{-1}}\; . 
\]
Denote $d_i:=\lambda_i^{-1}$ and $s:=\det C=\det D=\prod_{i=1}^n\lambda_i$ and divide \eqref{eq:konvexitaetsbedfuerproblem} by $s>0$ to obtain \eqref{eq:aequivalentzukonvexbed}.
\end{proof}

\begin{lemma}
 \label{lem:unglfuerzweiteablnotwendig}
 Let $f\in C^2(\Reell_+,\Reell)$ and $f\circ\det$ be convex on $\PSym(n)$. Then $f''(s)\geq-\frac{n-1}{n}\frac{f'(s)}{s}\;\forall s\in\Reell_+$.
\end{lemma}
\begin{proof}
 According to Lemma \ref{lem:aequivalentzukonvexbed}, \eqref{eq:aequivalentzukonvexbed} holds for all $H\in \Sym(n)$ and $D^{-1}=\diag(d_1,...,d_n)$. 
 Let $s\in\Reell_+$, $k\in\Reell\setminus\{0\}$ and $H=k\cdot D^{-1}$, as well as $D^{-1}=\diag(s^{-\frac{1}{n}},...,s^{-\frac{1}{n}})$.
 \begin{alignat*}{2}
  0&\leq&k^2& \left(\left(f''(s) + \frac{f'\left(s\right)}{s}\right)\sprodukt{D^{-1}}{D^{-1}}^2 -                                                          \frac{f'(s)}{s}\sprodukt{(D^{-1})^2}{(D^{-1})^2}\right)\\
  &=&k^2& \left(\left(f''(s) + \frac{f'\left(s\right)}{s}\right)\tr^2(D^{-1})^2 - \frac{f'(s)}{s}\cdot\tr (D^{-1})^4\right)\\
  &=&k^2& \left(\left(f''(s) + \frac{f'\left(s\right)}{s}\right)\left(ns^{-2/n}\right)^2 - \frac{f'(s)}{s}\cdot ns^{-4/n}\right)
  =\;nk^2s^{-4/n} \left(nf''(s) + (n-1)\frac{f'\left(s\right)}{s}\right).\tag*{\qedhere}
 \end{alignat*}
\end{proof}

For any matrix $A$ let $\diag A$ be the matrix obtained from $A$ by setting all non-diagonal entries zero. Let $\diag_{\M}$ be the set of all $n\times n-\,$ diagonal matrices.

\begin{lemma}\label{lem:diagabsch}
  For all $P\in\diag_{\M}$ with non-negative entries only and all $A\in\M$ the following holds:
  \begin{gather}
   \sprodukt{P}{A}=\sprodukt{P}{\diag A}=:\sigma(P,A)\; ,\label{untitled5}\\
   \sprodukt{PA}{AP}\geqslant\sprodukt{P\diag A}{\diag A\;P}=:\widetilde{\sigma}(P,A),\label{untitled6}\\
   \sigma^2(P,A)\leqslant n\cdot\widetilde{\sigma}(P,A)\; .\label{untitled7}
  \end{gather} 
\end{lemma}
\begin{proof}
 Let $P=\diag(p_1,...,p_n),\;A=(a_{ij})_{i,j}$ and calculate $\sprodukt{P}{A}=\sum_{i=1}^n p_ia_{ii}=\sprodukt{P}{\diag A}$. Hence \Refnrohne{untitled5} holds. 
 Direct calculation of $PA$ and $PA^T$ yields
 \[\sprodukt{PA}{AP}=\tr\left( PAPA^T\right)=\sum_{i=1}^n p_i^2a_{ii}^2+\sum_{i=1}^n\sum_{k\neq i}p_ip_ka_{ik}^2\geqslant\sprodukt{P\diag A}{\diag A\, P},\]
 i.e. \Refnrohne{untitled6}. 
 For all $P\in\diag_{\M}$ and $A\in\M$, we have $\sigma^2(P,A)\leqslant n\cdot\widetilde{\sigma}(P,A)$.
  To see this, note that $P$ and $\diag A$ commute, i.e. $\widetilde{\sigma}(P,A)=\sprodukt{P\diag A}{P\diag A}=\| P\diag A\|^2$ 
  holds. By Cauchy-Schwarz-inequality this immediately implies
  \begin{align*}
   \sigma^2(P,A)=\sprodukt{P}{\diag A}^2=\sprodukt{P\diag A}{\Id}^2 \leqslant \| P\diag A\|^2\cdot\|\Id\|^2=n\cdot\widetilde{\sigma}(P,A)    \; .\tag*{\qedhere}
  \end{align*}
\end{proof}


\begin{lemma}
 \label{lem:hinreichend}
 \Refnrohne{eq:dgluglfuerproblem} is sufficient for the convexity of $f\circ\det$.
\end{lemma}
\begin{proof}We will show \Aequizukonvexbedingung. To this end, let
 \(
  H\in\Sym(n)\,\text{ and }\; D^{-1}=\diag(d_1,...,d_n)
 \),
 where $d_1,...,d_n\in\Reell_+$ and $s:=(d_1\cdot\ldots\cdot d_n)^{-1}=\det D$ arbitrary. Then $P:=D^{-1}$ and $A:=H$ satisfy all assumptions of the previous lemma.  Using the notation from lemma \ref{lem:diagabsch}, we can, \obda, assume $\sigma(D^{-1},H)\neq 0$, because otherwise   \Refnrohne{eq:aequivalentzukonvexbed} becomes trivial by the assumption $f'\leqslant 0$. We denote $\sigma=\sigma(D^{-1},H)$ and $\widetilde\sigma=\widetilde\sigma(D^{-1},H)\leqslant \sprodukt{D^{-1}H}{HD^{-1}}$ by \eqref{untitled6}. Using $\frac{f'(s)}{s}\leq 0$ and $f''(s)+\frac{n-1}{n}\frac{f'(s)}{s}\geq 0$ by \eqref{eq:dgluglfuerproblem} and $1-\frac{\widetilde\sigma}{\sigma^2}\leq\frac{n-1}{n}$  
 we obtain \eqref{eq:aequivalentzukonvexbed} from
  \begin{align*}
  &\left(f''(s) + \frac{f'\left(s\right)}{s}\right){\sprodukt{D^{-1}}{H}}^2 -\,\frac{f'(s)}{s}\; \sprodukt{D^{-1}H}{HD^{-1}}
  \geqslant \left(f''(s) + \frac{f'\left(s\right)}{s}\right)\cdot\sigma^2- \frac{f'(s)}{s}\cdot \widetilde{\sigma}\\
  &=\sigma^2\cdot\Biggl[\; f''\left(s\right) +\frac{f'\left(s\right)}{s}\cdot\,\left(1 - \frac{\widetilde{\sigma}}{\;\sigma^2\;}\right)\;           \Biggr]
  \geqslant \sigma^2\cdot\left(f''\left(s\right) + \frac{n-1}{n}\cdot\frac{f'\left(s\right)}{s}                     \right)\;\geqslant\, 0\; .\tag*{\qedhere}
 \end{align*}
\end{proof}

\section{Solutions to the differential inequalities }\label{sec:lsgnderdglugl}

In this section we are interested in the possible shape of the functions that satisfy \Dglugl. To make calculations and figures more concrete, we restrict ourselves to the case $n=3$.

\begin{lemma}\textsc{(Linear ODE)}\label{theo:linearedgl}
  The linear initial value problem
  \begin{align*}  
   Ly:=y'+g(x)y=0, \; y(\xi)=\eta\tag{LIVP}\label{eq:LAWP}
  \end{align*}
   on $J=\Reell_+$ and where $g(x)=\frac{2}{3x}$ has one and only one solution.
\end{lemma}

To find solutions to $Lf'\geq 0$ under the additional constraint $y=f'\leq 0$ (which is equivalent to $\eta\leq 0$ because $f'\equiv 0$ is a solution)
we consider the ``limiting case'': 

\begin{lemma}\textsc{(Limiting case for \Dglugl)}\label{lem:grenzfallderdglugl}
 The solutions to
  \begin{align*}
   \flimit''(s)+\frac{2}{3s}\cdot \flimit'(s)=0\quad\text{and}\quad \flimit'(s)\leqslant 0\qquad\alle s\in\Reell_+
  \end{align*}
 are given by $\flimit\colon \Reell_+\to\Reell,\; s\mapsto c\cdot s^{1/3} + d$, where $c\leqslant 0\, ,\;d\in\Reell.$
 \end{lemma}
\begin{proof}
 Separation of variables gives the unique solution of \Lawp\  for $\xi>0\geqslant\eta$: 
 \begin{equation}\label{eq:ylimit}
  \ylimit(x)=\eta\cdot\exp\left(-\int_{\xi}^{x}\frac{2}{3t}\,dt\right) =                                                                               \eta\cdot\exp\left(-\frac{2}{3}\ln\frac{x}{\xi}\right) = \eta\,\xi^{2/3}\cdot x^{-2/3}.
 \end{equation}
 Because $\eta\,\xi^{2/3}\leqslant 0$, we have $\ylimit\leqslant 0$, hence \Refnrohne{eq:dgluglfuerproblem}.
 The claim follows by integration of $\flimit'=\ylimit$ with $c:=3\eta\,\xi^{2/3}$ and constant $d$.
\end{proof}

If we consider an interval adjacent to $\xi$ on the left hand side, i.e. $\widetilde{J}:=[\xi-a,\xi]$, the conditions for a function $y$ to be a sub- (or super)solution to $y'=F(x,y),\, y(\xi)=\eta$ are
  \[
   v'\,\begin{Bmatrix}>\\ \geqslant\end{Bmatrix}\, F(x,v)\text{ in }\widetilde{J}\, ,\quad v(\xi)\leqslant\eta \qquad \left(\text{or }     w'\,\begin{Bmatrix}<\\ \leqslant\end{Bmatrix}\, F(x,w)\text{ in }\widetilde{J}\, ,\quad w(\xi)\geqslant\eta \text{ respectively}\right)\; ,
  \]  
where (cf. \Lawp) 
\(
 F_{2/3}(x,y):=-\,\frac{2}{3x}\cdot y\in C^{\infty}(\Reell_+\times\Reell)\; 
\)
yields 
\begin{folgerung}\label{lem:oberunterlemma}
  Let $y$ be differentiable in $\Reell_+$,  $\xi > 0\geqslant\eta$. Then
  \[
   y'\,\begin{Bmatrix}>\\ \geqslant\end{Bmatrix}\, F_{2/3}(x,y)\, ,\; y(\xi)\geqslant\eta \Impl 
      \; y(x)\,\begin{Bmatrix}>\\[0.5ex] \geqslant\end{Bmatrix}\, \ylimit(x)=\eta\,\xi^{2/3}\cdot x^{-2/3}\;\text{ on }\,                               \begin{Bmatrix}\,(\xi,\infty)\,\\[0.5ex]\,[\xi,\infty)\,\end{Bmatrix}\; .
  \]
  Analogously:
  \[
   y'\,\begin{Bmatrix}>\\ \geqslant\end{Bmatrix}\, F_{2/3}(x,y)\, ,\quad y(\xi)\leqslant\eta \quad \Impl \quad                               y(x)\,\begin{Bmatrix}<\\[0.5ex] \leqslant\end{Bmatrix}\, \ylimit(x)\;\text{ on }\,                                                       \begin{Bmatrix}\,(0,\xi)\,\\[0.5ex]\,(0,\xi]\,\end{Bmatrix}\; .
  \]
\end{folgerung}

By these considerations, we obtain information on the qualitative shape of the solutions to \eqref{eq:dgluglfuerproblem}. (At first discussing the shape of $y=f'$.)
Note that 
to fulfill $y\leq 0$, in Lemma \ref{lem:oberunterlemma} also $0\geqslant y(\xi)\geqslant\eta$ must be satisfied. 
For $\eta =0$, $\ylimit\equiv 0$ is the unique solution and intersects $y$ in $\xi$. ($0\geq y(\xi)\geq 0$)\\
For $0\geqslant y(\xi)>\eta$, $y$ and $\ylimit$ with initial value $\ylimit(\xi)=y(\xi)$ intersect in $\xi$. 
Hence we can consider $y(\xi)=\eta=\ylimit(\xi)$ only. Then
\begin{align*}
  y'\,\begin{Bmatrix}>\\ \geqslant\end{Bmatrix}\, -\,\frac{2}{3x}\cdot y=F_{2/3}(x,y)\, ,\quad y(\xi)=\eta
\end{align*}
implies
\begin{align*}
  y\,\begin{Bmatrix}>\\[0.5ex] \geqslant\end{Bmatrix}\, \ylimit\;\text{ on }\,\begin{Bmatrix}\,(\xi,\infty)\,\\[0.5ex]\,                   [\xi,\infty)\,\end{Bmatrix} \quad\text{and}\quad                                                                                              y\,\begin{Bmatrix}<\\[0.5ex] \leqslant\end{Bmatrix}\, \ylimit\;\text{ on }\,\begin{Bmatrix}\,(0,\xi)\,\\[0.5ex]\,                        (0,\xi] \,\end{Bmatrix}\;\,.
\end{align*}
Additionally, the graphs of solutions $\ylimit$ contain the points $(1,\eta\,\xi^{2/3})$. Hence there is no need to consider initial values different from $y(1)=\eta=\ylimit(1)$ for $\eta\leqslant 0$.	 

Therefore all (derivatives $y$ of) solutions to \Dglugl\ qualitatively have the shape of the dashed line. ($y\leqslant 0$, $y'=f''\geqslant 0$, furthermore $f$ and $\flimit$ have the same slope in $1$). In $(0,1)$, however, $f$ decreases more rapidly than $\flimit$, in $(1,\infty)$ less rapidly.

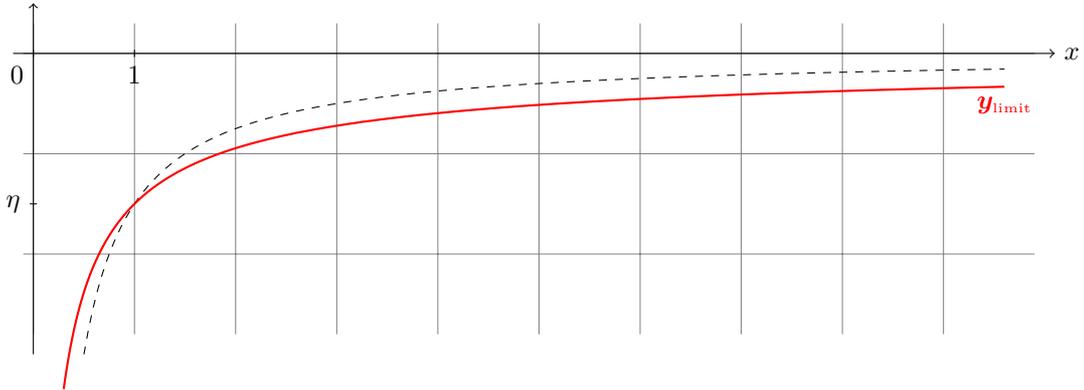
\begin{center}
\begin{tikzpicture}[scale=1.33,samples=200]
\draw[ultra thin,color=gray] (-0.1,-2.8) grid (9.9,0.3);
\draw[->] (-0.2,0) -- (10.1,0) node[right] {$x$};
\draw[->] (0,-3) -- (0,0.5) node[above] {};
\draw (1pt,-1.5 cm) -- (-1pt,-1.5 cm) node[anchor=east] {$\eta$};
\draw (0cm,1pt) -- (0cm,-1pt) node[anchor=north east] {$0$};
\draw (1cm,1pt) -- (1cm,-1pt) node[anchor=north] {$1$};
\draw[dashed,color=black] plot[domain=0.5:9.605] (\x,{(-1.5)*(1/\x)});
\draw[thick,color=red] plot[domain=1:9.605] (\x,{(-1.5)*(1/\x^(2/3))}) node[below] {$\boldsymbol{\ylimit}$};
\draw[thick,color=red] plot[domain=0.3:1] (\x,{(-1.5)/(1/\x^(2/3))}) node[below] {};
\end{tikzpicture}
\end{center}

The question now is: Are there other solutions of \eqref{eq:dgluglfuerproblem}?
Note that e.g. the attempt to find a solution by solving $y'=\widetilde{F}(x,y):=F_{2/3}(x,y)+\varepsilon\;$ for some positive $\varepsilon$ leads to solutions that satisfy $f\leq 0$ in a bounded neighbourhood of $\xi$ only and not on the whole of $\Reell_+$.
However, $\widetilde{F}(x,y):= F_{(2/3)\,+\,a}(x,y)=-\,\frac{y}{x}\cdot\left(\frac{2}{3}+a\right)=-\frac{3a+2}{3x} y$ for $a\geq 0$ gives  
$ y_a(x)=\eta\cdot\exp\left(-\int_1^x \frac{3a+2}{3t}\,dt \right)=\eta\cdot x^{-\left(\frac{2}{3}+a\right)}$ as solution to
\begin{align*}
   y_a'=-\,\frac{3a+2}{3x}\, y_a 
                 \quad \text{and}\quad & y_a(1)=\eta\leqslant 0
\end{align*}
and hence we obtain the following family of solutions to \eqref{eq:dgluglfuerproblem}:

\begin{lemma}\textsc{(Family of solutions)}\label{theo:loesungsscharfuerproblem}
    For arbitrary $c\leqslant 0,\,d\in\Reell$, $a\in [0,\infty)$, the family of functions that is defined on $\,\Reell_+$ by
    \begin{align*}
     f_a(s):=\begin{cases}d+c\cdot s^{\frac{1}{3}-a}   &,\,\text{ for }a\in[0,1/3)\\
                          d+c\cdot\ln s                &,\,\text{ for }a=1/3\\
                          d-c\cdot s^{\frac{1}{3}-a}   &,\,\text{ for }a\in(1/3,\infty)   
     \end{cases}
    \end{align*}
    has the property that $f_a\circ \det\colon \PSym(3)\to\Reell,\;C \mapsto f_a(\det C)$ is convex. 
    \end{lemma}
    \begin{bemerkung}
      Although this condition is not necessary for the convexity of $f\circ\det$, at least qualitatively all solutions to \eqref{eq:dgluglfuerproblem} have the shape indicated in the graph below. (As we discussed in this section.)
   \end{bemerkung}
    
   
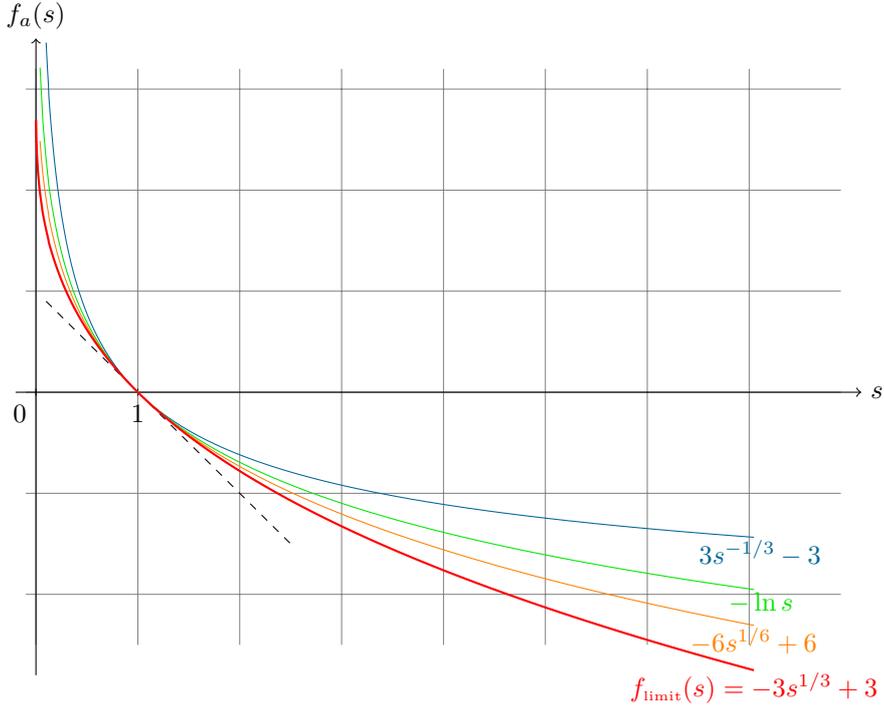
\begin{tikzpicture}[scale=1.34,samples=240]
   \draw[ultra thin,color=gray] (-0.1,-2.5) grid (7.9,3.2);
   \draw[->] (-0.2,0) -- (8.1,0) node[right] {$s$};
   \draw[->] (0,-2.8) -- (0,3.5) node[above] {$f_a(s)$};
   \draw (0cm,1pt) -- (0cm,-1pt) node[anchor=north east] {$0$};
   \draw (1cm,1pt) -- (1cm,-1pt) node[anchor=north] {$1$};
   \draw[dashed,color=black] plot[domain=0.1:2.5] (\x,{-\x+1});
   \draw[color=green!40!blue] plot[domain=1:7.05] (\x,{3/(\x^(1/3))-3}) node[below,xshift=0.8mm,yshift=0.9mm] {$3s^{-1/3}-3$};
   \draw[color=green!90!black] plot[domain=0.04:7.05] (\x,{-ln(\x)}) node[below,xshift=0.9mm,yshift=0.9mm] {$-\ln s$};
   \draw[color=orange] plot[domain=1:7.05] (\x,{(-6)*(\x^(1/6))+6}) node[below,yshift=0.9mm] {$-6s^{1/6}+6$};
   \draw[thick,color=red] plot[domain=1:7.05] (\x,{(-3)*(\x^(1/3))+3}) node[below,yshift=1mm] {${\flimit(s)=-3s^{1/3}+3}$};
   \draw[color=green!40!blue] plot[domain=0.1:1] (\x,{3*(\x^(1/3))-3}) node[below,xshift=0.8mm,yshift=0.9mm] {};
   \draw[color=orange] plot[domain=0.04:1] (\x,{(-6)/(\x^(1/6))+6}) node[below,yshift=0.9mm] {};
   \draw[thick,color=red] plot[domain=0.001:1] (\x,{(-3)/(\x^(1/3))+3}) node[below,yshift=1mm] {};

   \end{tikzpicture}

\section*{Acknowledgement}
We thank John Ball for his interest in this result.

\bibliographystyle{amsplain}
\providecommand{\bysame}{\leavevmode\hbox to3em{\hrulefill}\thinspace}
\providecommand{\MR}{\relax\ifhmode\unskip\space\fi MR }
\providecommand{\MRhref}[2]{%
  \href{http://www.ams.org/mathscinet-getitem?mr=#1}{#2}
}
\providecommand{\href}[2]{#2}

\end{document}